\documentclass[a4paper]{article}

\usepackage{comment}
\usepackage[margin=3cm]{geometry}
\usepackage{algorithm2e}
\usepackage[pagebackref]{hyperref}
\renewcommand*\backref[1]{\ifx#1\relax \else ($\uparrow$ #1) \fi}
\usepackage{tikz}
\usetikzlibrary{cd}
\usepackage{enumerate}
\usepackage{amsmath}
\usepackage{amssymb}
\usepackage{amsthm}
\usepackage{mathtools}
\newtheorem{theorem}{Theorem}
\newtheorem{definition}[theorem]{Definition}
\newtheorem{proposition}[theorem]{Proposition}

\title{A potentialization algorithm for games with applications to multi-agent learning in repeated games}

\date{\today}

\author{Philipp Lakheshar\footnote{\texttt{philipp.lakheshar@hotmail.com}}, Sharwin Rezagholi\footnote{\texttt{sharwin.rezagholi@technikum-wien.at}}\\
\small University of Applied Sciences Technikum Wien \\
\small Vienna, Austria
}

\begin{document}
\maketitle

\begin{abstract}
\noindent We investigate an algorithm that assigns to any game in normal form an approximating game that admits an ordinal potential function. Due to the properties of potential games, the algorithm equips every game with a surrogate reward structure that allows efficient multi-agent learning. Numerical simulations using the replicator dynamics show that 'potentialization' guarantees convergence to stable agent behavior.

%\noindent \textbf{MSC 2020}: 91A06 91A14 Game theory, 93D99 Stability of control systems, 68W05 Algorithms in computer science, 37N40 Applications of dynamical systems.
\end{abstract}

\section{Introduction}
Reinforcement learning is the machine learning paradigm in which an agent learns a desired behavior through interaction with its environment. Recent successes in the domain relied on the combination of reinforcement-learning-principles and deep learning. Many potentially relevant scenarios involve the participation of more than one agent. The naive approach to multi-agent reinforcement learning (MARL) is to consider all involved agents as self-motivated learners within a shared environment, that is to employ independent single-agent algorithms. In doing so, the interactive nature of the problem, i.e. that the rewards of individual agents are conditioned on the joint behavior of all agents, is ignored. When applied in multi-agent settings, single-agent algorithms may fail to converge or may converge to undesirable stable states. The complexity of multi-agent learning can be attributed to the inherent nonstationarity of the problem: Agents learn simultaneously and independently and thus continuously change each other's environment as their behavior evolves \cite{zhang2021multi}. Another naive approach is to consider MARL as the single-agent problem of a meta-agent that controls the actions of all agents. Aside from those cases in which the resulting combinatorial explosion would render the problem intractable, many practically relevant cases inherently require the coordination of autonomous agents.

This paper is concerned with multi-agent tasks in the form of repeated games in normal form. Ordinal potential games \cite{monderer1996potential} admit a function defined on the space of joint actions that simultaneously encodes the incentives of all players. In a sense these games are strategically equivalent to common interest games. It is well-known that a wide class of myopic learning methods are guaranteed to converge to Nash equilibria in ordinal potential games \cite{young1993, young2004}. The central idea of the present paper is to replace the reward functions of agents with reward functions that correspond to an ordinal potential game while retaining the incentive structure of the original game as faithfully as possible. This enables the extension of the learnability properties of ordinal potential games to a more general domain of MARL-problems at the price of some distortion of rewards. To gain insight into the effects of our method, we randomly generate $(10 \times 10)$- and $(4 \times 4 \times 4)$-games in normal form, apply the potentialization algorithm, and use the replicator equation \cite{cressman2014replicator} to simulate multiple simultaneous Q-learners \cite{watkins1992q}. We compare the dynamics of learning in the original games and in their potentialized versions and find that fast convergence is indeed guaranteed by potentialization while reward distortions are minimal.

\section{Games and potentials}

\begin{definition}[Finite game in normal form]
A finite game of $n$ players in normal form $G$ is the tuple
$$
G = \left( A = \prod_{i=1}^n A_i , \big\{ u_i: A \to \mathbb{R} \big\}_{i=1}^n \right) ,
$$
where $A_i$ is the finite action set of player $i \in \{1,\dots,n\}$ and $u_i$ denotes the utility function of player $i$. 
\end{definition}

The strategic incentives in a game are encoded in its deviation graph. We slightly abuse notation and write $a_{-i}$ to denote an element of the set $A_{-i} \coloneqq \prod_{j \in \{1,\dots,n\} \setminus \{i\}} A_j$. We denote a stochastic action by player $i$ as $\pi_i \in \Delta(A_i)$, a joint stochastic action as $\pi \in \Delta(A)$, and a joint stochastic action of all players but $i$ as $\pi_{-i} \in \Delta(A_{-i})$.

\begin{definition}[Deviation graph]
The deviation graph of the game $G$ is the weighted directed graph $\Gamma = (V, E, w)$ with vertex set $V=A$, edge set
$$
E = \left\{ a,a' \in A\times A: a=(a_i, a_{-i}), a'=(a'_i, a_{-i}), a_i \neq a'_i \right\} ,
$$
and weighting function $w: E \to \mathbb{R}$ given by
$$
w(a,a') = u_i \big( (a'_i, a_{-i}) \big) - u_i \big( (a_i, a_{-i}) \big) .
$$
The nonnegative deviation graph $\Gamma_0 = (V, E_0, w)$ contains all vertices of $\Gamma$ but only those of its edges with nonnegative weights:
$$
E_0= \left\{ e \in E: w(e) \geq 0 \right\} .
$$
\end{definition}
We extend edge-weights additively to directed paths; the weight of the directed path $p = (e_1, \dots, e_m)$ in a weighted graph equals $w(p) = \sum_{i=1}^m w(e_i)$.

\begin{definition}[Ordinal potential \cite{monderer1996potential}]
The function $\Phi: A \to \mathbb{R}$ is an ordinal potential for the game $G$ if
$$
\Phi(a'_i, a_{-i}) - \Phi(a_i, a_{-i}) > 0 \iff u_i \big( (a'_i, a_{-i}) \big) - u_i \big( (a_i, a_{-i}) \big) > 0 .
$$
A game is called an ordinal potential game if it admits an ordinal potential function.
\end{definition}

Note that setting $u_i (a) = \Phi(a)$ for all $i$, that is changing every player's reward to the value of the ordinal potential function, creates a game whose nonnegative deviation graph has the same topology as the one of the original game, only its weights may change. Therefore a game with ordinal potential $\Phi$ and the common interest game where $u_i (a) = \Phi(a)$ have the same pure Nash equilibria \cite{monderer1996potential}. All potential games have at least one pure Nash equilibrium \cite{monderer1996potential} (an equilibrium in nonstochastic individual actions). This follows immediately from the characterization of the class of ordinal potential games by a property of their deviation graphs: A game admits an ordinal potential if and only if its deviation graph is devoid of weak improvement cycles.

\begin{definition}[Weak improvement cycle]
A weak improvement cycle in a deviation graph is a directed path $p=(e_1, \dots, e_m)$ where $e_m = e_1$, such that $w(e_t) \geq 0$ for all $t \in \{1,\dots,m\}$ and $w(p) > 0$.
\end{definition}

\begin{theorem}[Voorneveld and Nolde \cite{voorneveld1997}]
A game admits an ordinal potential function if and only if its deviation graph is devoid of weak improvement cycles.
\end{theorem}

The condition of Voorneveld and Nolde can be stated as: If there is a cycle, then its weight must be zero. It is known that a wide class of multi-agent learning processes converge to pure Nash equilibria in games with ordinal potentials (see \cite{young1993} and \cite[Sections 6.4-6.5, Footnote 6]{young2004}). Potential games thus have two desirable properties for multi-agent learning: They are learnable by myopic individual learning rules, and they possess Nash equilibria in deterministic actions.

\section{Potentialization algorithm}

We proceed in three steps: (i) We calculate an adjusted graph on the basis of the nonnegative deviation graph; (ii) We calculate an ordinal potential function using this adjusted graph, whereby we crucially employ the concept of condensing graphs to their strongly connected components; (iii) We define a common interest game using the obtained ordinal potential function. The method is summarized in Algorithm \ref{algo}.
\begin{definition}[Strongly connected component]
Consider a directed graph $(V, E)$. The vertices $C \subseteq V$ are a strongly connected component of the graph if and only if there exists a directed path from $v$ to $v'$ for any two nonidentical vertices $v,v' \in C$.
\end{definition}
All edges within a strongly connected component of $\Gamma_0$ must have edge weight equal to zero to fulfill the Voorneveld-Nolde condition, which may be equivalently stated as: A game admits an ordinal potential function if and only if all edges within the strongly connected components of its nonnegative deviation graph have zero weight. We compute the condensation of $\Gamma_0$, slightly generalizing the notion of condensation to the weighted case. Recall that the condensation of an unweighted graph is the quotient graph with respect to strongly connected components.
\begin{definition}[Condensation of weighted directed graphs]
The condensation of the weighted directed graph $\big( V, E, w \big)$ is the condensation of its underlying unweighted graph $\big( V, E\big)$ equipped with the weight function
$$
(C, K) \mapsto \max_{v, v'} \left\{ w \big( (v,v') \big) : v \in C, v' \in K \right\} ,
$$
where $C,K$ denote strongly connected components and $v, v'$ denote vertices.
\end{definition}
The computation of strongly connected components can be achieved by algorithms of complexity $O(|V|+|E|)$ \cite{tarjan1972, sharir1981}. Our method proceeds on the basis of the condensation $\Gamma^*$ of the graph $\Gamma_0$. The potential function must be constant when restricted to any strongly connected component. A topological sorting of a directed graph is a linear ordering of its vertices such that, whenever there is an edge from vertex $v$ to vertex $w$, the vertex $v$ precedes $w$ in the topological sorting. A graph admits a topological sorting if and only if it does not have directed cycles. Therefore the condensation of any game's nonnegative deviation graph admits a topological sorting. Algorithms for topological sorting of complexity $O(|V|+|E|)$ are available \cite{kahn1962}. Our method consists of topologically sorting the vertices of the condensation $\Gamma^*$ and separately defining the values of the potential function for each vertex of $\Gamma^*$. For each of the vertices of $\Gamma^*$ we set the values of the potential function to zero whenever the vertex has no incoming edges. We then proceed along the topological sorting whereby we consider for every incoming edge the sum of the respective edge weight with the value of the potential function on the origin of the respective edge; the largest of the considered values is assigned as the value of the potential function. The method is best understood by example (Figure \ref{eg}).
\begin{figure}
\begin{center}
\begin{tikzcd}
a \ar{r}{1} \ar[bend left=40pt]{rr}{3} & b \ar{r}{2} \ar[bend right=40pt]{rr}{4} & c \ar{r}{5} & d \ar{r}{1} & e
\end{tikzcd}
\end{center}
\begin{eqnarray*}
\Phi(a) &=& 0 \\
\Phi(b) &=& \Phi(a) + 1 = 1 \\
\Phi(c) &=& \max \{ \Phi(a)+3, \Phi(b)+2 \} = 3 \\
\Phi(d) &=& \max \{ \Phi(c)+5, \Phi(b)+4 \} = 8 \\
\Phi(e) &=& \Phi(d) + 1 = 9
\end{eqnarray*}
\caption{Construction of the potential function.}
\label{eg}
\end{figure}
The common interest game where the utility functions of all players are given by the computed potential function is the potentialized game. All steps of the method are summarized in Algorithm \ref{algo}.
\begin{algorithm}
\KwData{Game $G$ in normal form.}
\KwResult{Potentialized approximation to $G$.}
Compute the nonnegative deviation graph $\Gamma_0$\;
Compute the condensation $\Gamma^*$ of $\Gamma_0$\;
Topologically sort the vertices of $\Gamma^*$ to obtain the sorting $\{v_1, \dots, v_k\}$\;
\For{$j \in \{1,\dots,k\}$}{
\eIf{$v_j$ has no incoming edges}{
$\Phi \big( [v_j] \big) \leftarrow 0$\;
}{
$\Phi \big( [v_j] \big) \leftarrow \max_{v'} \left\{ \phi(v') + w(e) : v' \xrightarrow{e} v_j \right\}$\;
}
}
Set $u_i (a) = \Phi(a)$ for all $a \in A$ and $i \in \{1,\dots,n\}$\;
\caption{Potentialization algorithm. We use the notation $[v]$ to denote the equivalence class of vertex $v$ with respect to being situated in the same strongly connected component.}
\label{algo}
\end{algorithm}

All steps of the method can be computed using algorithms of complexity $O(|V|+|E|)$. Let $k \coloneqq \max_i |A_i|$. We have $|V| \leq k^n$, and $|E| \leq (k-1) k n \leq k^2 n$. Therefore the complexity of the entire algorithm is $O(k^n)$. It must be acknowledged that this poses serious limits on the scalability of the algorithm, especially with respect to the number of players.

\section{Properties of the algorithmic result}

Our method is optimal insofar as it computes the smallest function $\Phi$ that is constant on strongly connected components of the original nonnegative deviation graph while increasing along its edges by at least the edges' original weights. The latter property is desirable, since any reasonable single-agent learning system would finally make the possible deviator recognize the deviation gain. By making the maximal deviation gain the common deviation gain to all players, the inevitability of deviation is immediately apparent to all players, which increases the speed of the learning process.

The algorithm turns any game into a game with at least one nonstochastic equilibrium. The algorithm may turn action profiles which are not equilibria into equilibria. This is vividly illustrated by the family of games
$$
\begin{bmatrix}
h,0 & 0,1 \\
0,1 & 1,0
\end{bmatrix}
$$
where $h > 1$. In the unique Nash equilibrium the column player uses their first action with probability $1/(h+1)$ and the row player uses each of their actions with equal probability. The sum of expected utilities in equilibrium equals $1 + (h-1) / (2 + 2h)$. The potentialized game is
$$
\begin{bmatrix}
0,0 & 0,0 \\
0,0 & 0,0
\end{bmatrix}
$$
for any value $h>1$; a game in which any (stochastic) behavior is a Nash equilibrium. In this example potentialization may lead to an outcome where the sum of expected utilities decreases by $(h-1) / (2 + 2h)$, a term that can be arbitrarily large, depending on $h$. This shows that large utility distortions are possible. But such effects can not be precluded unless one is willing to force players to act against their preferences, that is to break a cycle in the deviation graph by 'reversing an edge'. A less dramatic instance is Figure \ref{eg3}, although action profiles that involve weakly dominated actions are turned into equilibria. Even when potentialization increases the number of Nash equilibria, not all of these equilibria may be attractive for the learning dynamics. An example where potentialization increases the number of equilibria while still effectively undertaking equilibrium selection is Figure \ref{eg4}.
\begin{figure}
\begin{center}
Original game:

\vspace{10pt}

$
\begin{bmatrix}
3,3 & 2,0 & 1,0 \\
0,2 & 2,2 & 1,0 \\
0,1 & 0,1 & 1,1
\end{bmatrix}
$

\vspace{10pt}

Deviation graph:

\vspace{10pt}

\begin{tikzcd}[column sep=huge, row sep=huge]
\alpha & \bullet \ar[swap]{l}{3} & \bullet \ar[swap, leftrightarrow]{l}{0} \ar[swap, bend right=30pt]{ll}{3} \\
\bullet \ar{u}{3} \ar[leftrightarrow]{r}{0} & \beta \ar[bend left=30pt, leftrightarrow]{u}{0} & \bullet \ar[bend left=30pt]{l}{2} \ar[leftrightarrow, swap]{u}{0} \ar[swap, bend right=30pt]{ll}{2} \\
\bullet \ar[bend left=30pt]{uu}{3} \ar[leftrightarrow]{u}{0} \ar[swap, leftrightarrow]{r}{0} & \bullet \ar{u}{2} \ar[swap, bend right=30pt]{uu}{2} & \gamma \ar[leftrightarrow, swap]{u}{0} \ar[leftrightarrow]{l}{0} \ar[swap, leftrightarrow, bend right=30pt]{uu}{0} \ar[leftrightarrow, bend left=30pt]{ll}{0}
\end{tikzcd}

\vspace{10pt}

Condensation:

\vspace{10pt}

\begin{tikzcd}[column sep=huge, row sep=huge]
\alpha & \bullet \ar[swap]{l}{3}
\end{tikzcd}

\vspace{10pt}

Potentialized game:

\vspace{10pt}

$
\begin{bmatrix}
3,3 & 0,0 & 0,0 \\
0,0 & 0,0 & 0,0 \\
0,0 & 0,0 & 0,0
\end{bmatrix}
$
\end{center}
\caption{Coordination game. The pure Nash equilibria of the original game are Pareto-ordered: $\alpha > \beta > \gamma$. The second and the third action of both players is weakly dominated. There is a cycle containing all vertices but $\alpha$. All pure equilibria are retained and two additional pure equilibria are created.}
\label{eg3}
\end{figure}
\begin{comment}
\begin{figure}
\begin{center}
Orig. game:
$
\begin{bmatrix}
2,1 & 0,1 \\
1,0 & 0,0
\end{bmatrix}
$
, Dev. graph:
\begin{tikzcd}
\alpha \ar[leftrightarrow]{r}{0} & \beta \ar[leftrightarrow]{d}{0} \\
\bullet \ar{u}{1} & \gamma \ar[leftrightarrow]{l}{0}
\end{tikzcd}
, Pot. game:
$
\begin{bmatrix}
0,0 & 0,0 \\
0,0 & 0,0
\end{bmatrix}
$
\end{center}
\caption{The three pure Nash equilibria of the original game are Pareto-ordered: $\alpha > \beta > \gamma$. The second action of the row player is weakly dominated in the original game.}
\label{eg2}
\end{figure}
\end{comment}
\begin{figure}
\begin{center}
Original game:

\vspace{10pt}

$
\begin{bmatrix}
0,1 & 1,2 & 0,0 \\
1,1 & 0,0 & 0,0 \\
0,0 & 0,0 & 1,1
\end{bmatrix}
$

\vspace{10pt}

Deviation graph:

\vspace{10pt}

\begin{tikzcd}[column sep=huge, row sep=huge]
\bullet \ar[swap, leftrightarrow, bend right=30pt]{dd}{0} \ar[swap]{d}{1} \ar{r}{1} & \alpha & \bullet \ar[swap]{l}{2} \ar[swap, bend right=30pt]{ll}{1} \ar[leftrightarrow]{d}{0} \ar[bend left=30pt]{dd}{1} \\
\beta & \bullet \ar[bend left=30pt]{u}{1} \ar{l}{1} \ar[bend right=30pt, leftrightarrow, swap]{r}{0} \ar[swap, leftrightarrow]{d}{0} & \bullet \ar{d}{1} \ar[bend right=30pt]{ll}{1} \\
\bullet \ar{u}{1} \ar[bend right=30pt, swap]{rr}{1} & \bullet \ar[swap, bend right=30pt]{uu}{1} \ar[swap]{r}{1} \ar[leftrightarrow]{l}{0} & \gamma
\end{tikzcd}

\vspace{10pt}

Condensation:

\vspace{10pt}

\begin{tikzcd}
{} & \alpha & {} \\
\beta & \bullet \ar{u}{2} \ar{l}{1} \ar{dr}{1} & {} \\
{} & {} & \gamma
\end{tikzcd}

\vspace{10pt}

Potentialized game:

\vspace{10pt}

$
\begin{bmatrix}
0,0 & 2,2 & 0,0 \\
1,1 & 0,0 & 0,0 \\
0,0 & 0,0 & 1,1
\end{bmatrix}
$
\end{center}
\caption{Game from Voorneveld and Nolde (1997). The vertices denoted by $\bullet$ form a cycle. All three equilibria are retained but the Pareto-best original equilibrium becomes preferable to both players.}
\label{eg4}
\end{figure}

\begin{proposition}
(i) If $a \in A$ is a strict Nash equilibrium of the original game, then it is a Nash equilibrium of the potentialized game.
(ii) If $a \in A$ is a Nash equilibrium of the potentialized game, then its constituent actions were not strictly dominated in the original game.
\end{proposition}
\begin{proof}
(i) Since $a$ is a strict Nash equilibrium, it is a sink of the original game's deviation graph. Since there are no outgoing edges, $a$ also forms a strongly connected component of the condensed graph. It is a sink of the condensed graph. Therefore $a$ is a Nash equilibrium of the potentialized game.
(ii) We prove the contrapositive statement: If $a$ includes an action that is strictly dominated in the original game, then $a$ is not a Nash equilibrium of the potentialized game. Let $a=(a_i, a_{-i})$ where $a_i$ is strictly dominated in the original game. Then there exists $a'_i \in A_i \setminus \{a_i\}$ such that
$$
u_i (a'_i, a'_{-i}) > u_i (a_i, a'_{-i})
$$
for all $a'_{-i} \in A_{-i}$. Since
$$
\Phi (a'_i, a'_{-i}) - \Phi (a_i, a'_{-i}) \geq u_i (a'_i, a'_{-i}) - u_i (a_i, a'_{-i}) > 0
$$
we conclude
$$
\Phi (a'_i, a_{-i}) > \Phi (a_i, a_{-i})
$$
which implies that deviation from $(a_i, a_{-i})$ benefits $i$ in the potentialized game, i.e. $(a_i, a_{-i})$ is not a Nash equilibrium.
\end{proof}

Two methods to assign an approximating exact potential game to a game in strategic form have been suggested by Candogan et al. \cite{candogan2011, candogan2013}. An exact potential function increases by exactly the deviation gain along an edge of the deviation graph. An exact potential is an ordinal potential, but the converse is not the case \cite{monderer1996potential}. The methods of Candogan et al. are numerically more faithful to the original utilities, being based on a distance function defined via the utilities, but may disrespect the direction of edges in the original deviation graph. It is for example possible that the orientation of an edge within a strongly connected component of the deviation graph is changed \cite[Example 1.1]{candogan2011}. While our approach may lead to large utility distortions its property of never 'forcing players to act contrary to their preferences' is meaningful in its own right.

\section{Numerical experiments}

The replicator equation was introduced to model the dynamics of population games in evolutionary biology \cite{taylor1978}. Although it appeared in the context of evolutionary theory, the replicator equation is formally connected to multi-agent learning (e.g. \cite{tuyls2006evolutionary}).

\begin{definition}[Replicator differential equation \cite{taylor1978}]
The replicator equation for a game in normal form is the differential equation 
$$
\frac{\dot {\pi_i}_k}{{\pi_i}_k} = \sum_{a_{-i} \in A_{-i}} u_i \Big( {a_i}_k, a_{-i} \Big) \pi_{-i} \Big( a_{-i} \Big) - \sum_{a \in A} u_i (a) \pi (a) ,
$$
where ${\pi_i}_k$ denotes the probability that agent $i$ takes action $k$.
\end{definition}
The replicator system is defined on $\Delta(A)$, the set of joint stochastic actions. The replicator equation states that the relative change of the probability that agent $i$ takes the action $k$ equals the difference between, on the one hand, the expected utility of the action for $i$ given the joint stochastic action of the other agents, and, on the other hand, the expected utility of the agent given the current joint stochastic action.

The replicator equation can be obtained from the dynamics of Q-learning agents in a repeated game \cite{kianercy2012dynamics}. Q-learning is a commonly used off-policy reinforcement learning algorithm \cite{watkins1992q, deepq}. Kianercy and Galstyan \cite{kianercy2012dynamics} consider Q-learning agents during repeated play of a game in normal form and show that the replicator equation is obtained after the following steps: (i) Taking the limit of continuous time by linear interpolation; (ii) Taking the limit of small experimentation; (iii) A mean-field simplification. We therefore use the replicator equation to simulate multi-agent learning. An ordinal potential function is a Lyapunov function for the replicator equation of the respective game (e.g. \cite[page 87]{laraki2019}). This is another argument in favor of potential games for multi-agent learning.

\paragraph{Random generation of games.}
We specified the number of agents $n$ and the cardinality of their action sets. Games were constructed by independently drawing values
$$
u_i (a) \sim \text{uniform} \left( \{1, \dots, |A|\} \right) .
$$
Before applying the replicator equation to any game $\big( A, \left\{ u_i \right\}_{i=1}^n \big)$ we jointly rescaled all rewards to the interval $[0,1]$ by setting
$$
u'_i (a) = \frac{u_i (a) - \min_{i,a} u_i (a)}{\max_{i,a} u_i (a) - \min_{i,a} u_i (a)}
$$
if $\min_{i,a} u_i (a) < \max_{i,a} u_i (a)$ and $u'_i(a)=0$ otherwise.

\paragraph{Experiments.}
At each step $t$ we calculated the difference between the current and the previous joint stochastic action. When the quantity
\begin{equation}\label{policy_distance}
\beta (t) = \max_i \| \pi_i (t) - \pi_i (t-1) \|_2
\end{equation}
did not exceed the value $10^{-9}$ for $10^3$ steps we declared convergence. We used the fourth-order Runge-Kutta method with the step size $10^{-2}$. The starting points $\pi_i( 0)$ for our simulations were independently drawn from uniform distributions on $\Delta (A_i)$. Since games without potential may never lead to convergence under the replicator equation, we stopped the replicator system on the original games latest at the number of steps that sufficed for the potentialized version to converge. We generated $1,000$ games of size $(4 \times 4 \times 4)$, that is games with three players and four actions per player, and $1,000$ games of type $(10 \times 10)$, that is games with two players and ten actions per player.

Our numerical experiments consist of the steps: (i) A game $G$ is randomly generated;
(ii) The potentialization algorithm is applied to $G$ to obtain the potentialized game $G_\Phi$; (iii) The game $G_\Phi$ is normalized and its replicator equation is simulated until it converges; (iv) The game $G$ is normalized and its replicator equation is simulated for as long as for the potentialized game.

\paragraph{Results.} Figures \ref{10_10} and \ref{4_4_4} show the progression of the policy variability measure $\beta$ (Equation \ref{policy_distance}). The figures also show the average instantaneous rewards of the agents when using the current policy in the original game
\begin{equation}\label{average_reward}
\rho (t) = \frac{1}{n} \sum_{i=1}^n \sum_{a \in A} u_i (a) \pi (t)_a .
\end{equation}
For the $(10 \times 10)$-games, $96.6\%$ of the potentialized instances converged within our time limit of $10^5$ steps, whereas $8.6\%$ of the original instances converged. Of the $(4 \times 4 \times 4)$-games, $90.6\%$ of the potentialized games converged, whereas $11.8\%$ of the original instances converged.
\begin{figure}
\begin{center}
\includegraphics[width=\textwidth]{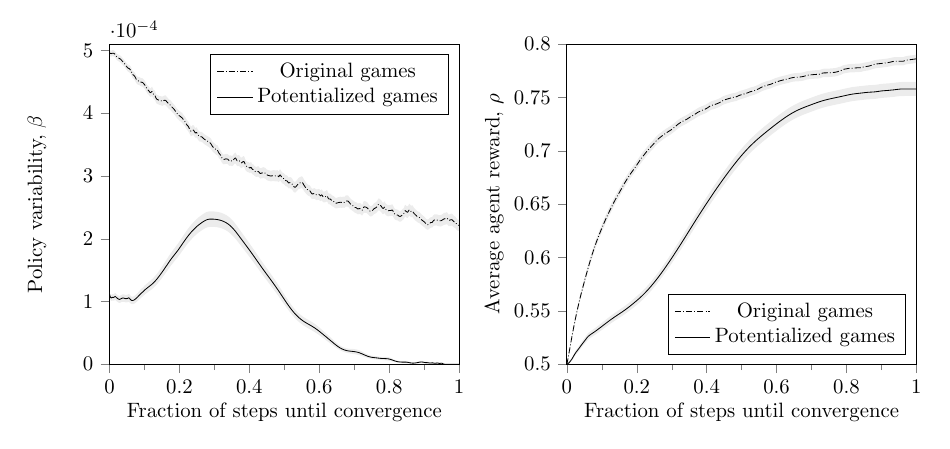}
\end{center}
\caption{Experimental results for $(10 \times 10)$-games ($\pm$ empirical standard deviation).}\label{10_10}
\end{figure}
\begin{figure}
\begin{center}
\includegraphics[width=\textwidth]{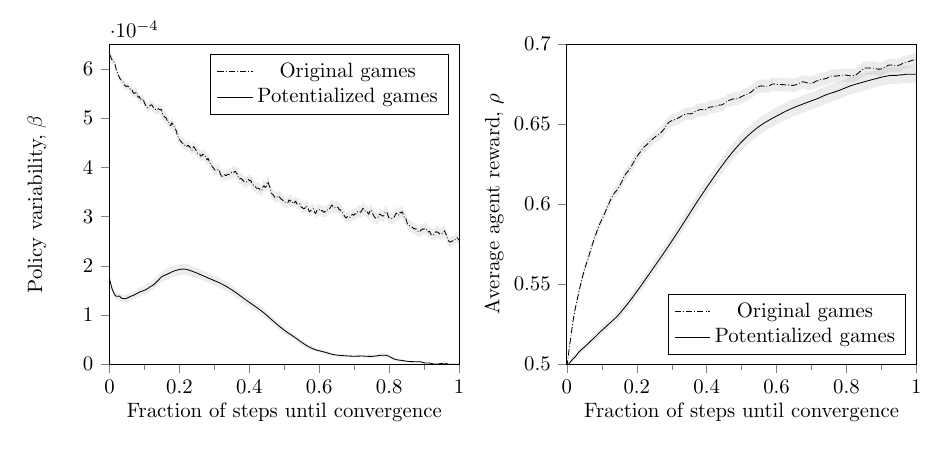}
\end{center}
\caption{Experimental results for $(4 \times 4 \times 4)$-games ($\pm$ empirical standard deviation).}\label{4_4_4}
\end{figure}
For the $(10 \times 10)$- as well as for the $(4 \times 4 \times 4)$-games the policy variability (Equation \ref{policy_distance}) of the potentialized games lies consistently below the level of the original games. While its evolution displays a consistently decreasing trend for the original games, there is an initial increase in the case of the potentialized games followed by a consistent decrease until the policy variability vanishes. The curve of the potentialized games is much smoother, hinting at greater stability. We have found that in many individual cases the original games exhibit periodical patterns and volatile trajectories. Concerning the average expected reward of agents (Equation \ref{average_reward}), the potentialized games exhibit slightly smaller average expected reward but are smoother. In the case of $(10 \times 10)$-games, the original games reach an average agent reward of $0.79$ at the end of the experiment, while the average reward in the potentialized games ends at $0.76$, which corresponds to $96.4\%$ of the average reward in the original games. For $(4 \times 4 \times 4)$-games, the curves end at $0.69$ for the original games and at $0.68$ for the potentialized games, which corresponds to $98.6\%$ of the average reward in the original games.

\section{Conclusion}

The present paper presented an algorithm that approximates a game in normal form by an ordinal potential game. Numerical evidence suggests that this potentialization method permits fast multi-agent learning. The price of stability consists of possible decreases in average agent reward. The method is promising for application in the context of reinforcement learning, either by exclusively using the potentialized reward structure, or by using the potentialized reward structure during a burn-in phase and to later replace it by the original reward structure. The latter usage speeds up the learning process on the 'common-interest-part' of the game and allows to quickly focus on the 'competitive part' of the game. Future work could investigate the extension of the present approach to multi-agent Markov decision processes, that is to stochastic games \cite{shapley1953}. The extension to stochastic games might proceed via the potentialization of their implicitly defined instantaneous continuation games.

%\paragraph{Author contribution.} SR conceptualized and supervised research. PL implemented the algorithm and undertook numerical experiments. PL and SR wrote the paper.

\bibliographystyle{plain}
\bibliography{Literatur}

@article{candogan2011,
author = {Ozan Candogan and Ishai Menache and Asuman Ozdaglar and Pablo A. Parrilo},
journal = {Mathematics of Operations Research},
number = {3},
pages = {474--503},
title = {Flows and Decompositions of Games: {H}armonic and Potential Games},
volume = {36},
year = {2011}
}

@article{candogan2013,
title = {Dynamics in near-potential games},
journal = {Games and Economic Behavior},
volume = {82},
pages = {66-90},
year = {2013},
doi = {https://doi.org/10.1016/j.geb.2013.07.001},
author = {Ozan Candogan and Asuman Ozdaglar and Pablo A. Parrilo}
}

@article{shapley1953,
author = {Lloyd Shapley},
title = {Stochastic Games},
journal = {Proceedings of the National Academy of Sciences},
volume = {39},
number = {10},
pages = {1095-1100},
year = {1953},
doi = {10.1073/pnas.39.10.1095}
}

@article{kahn1962,
author = {A.B. Kahn},
title = {Topological sorting of large networks},
year = {1962},
volume = {5},
number = {11},
doi = {10.1145/368996.369025},
journal = {Communications of the ACM},
pages = {558–562}
}

@article{tarjan1972,
author = {Robert Tarjan},
title = {Depth-First Search and Linear Graph Algorithms},
journal = {SIAM Journal on Computing},
volume = {1},
number = {2},
pages = {146-160},
year = {1972},
doi = {10.1137/0201010}
}

@article{sharir1981,
title = {A strong-connectivity algorithm and its applications in data flow analysis},
journal = {Computers and Mathematics with Applications},
volume = {7},
number = {1},
pages = {67-72},
year = {1981},
doi = {https://doi.org/10.1016/0898-1221(81)90008-0},
author = {M. Sharir}
}

@article{young1993,
 author = {H. Peyton Young},
 journal = {Econometrica},
 number = {1},
 pages = {57--84},
 title = {The Evolution of Conventions},
 volume = {61},
 year = {1993}
}

@book{young2004,
    author = {H. Peyton Young},
    title = {Strategic Learning and Its Limits},
    publisher = {Oxford University Press},
    year = {2004}
}

@book{laraki2019,
title={Mathematical Foundations of Game Theory},
year={2019},
author={Rida Laraki and Jerome Renault and Sylvain Sorin},
publisher={Springer}
}

@article{voorneveld1997,
title = {A Characterization of Ordinal Potential Games},
journal = {Games and Economic Behavior},
volume = {19},
number = {2},
pages = {235--242},
year = {1997},
doi = {https://doi.org/10.1006/game.1997.0554},
author = {Mark Voorneveld and Henk Norde}
}

@article{deepq,
    author = {Volodymyr Mnih and Koray Kavukcuoglu and David Silver and Andrei Rusu and Joel Veness and Marc Bellemare and Alex Graves and Martin Riedmiller and Andreas Fidjeland and Georg Ostrovski},
    title = {Human-level control through deep reinforcement learning},
    journal = {Nature},
    year = {2015},
    volume = {518},
    issue = {7540},
    pages = {529--533}
}

@article{taylor1978,
title = {Evolutionary stable strategies and game dynamics},
journal = {Mathematical Biosciences},
volume = {40},
number = {1},
pages = {145--156},
year = {1978},
doi = {https://doi.org/10.1016/0025-5564(78)90077-9},
author = {Peter Taylor and Leo Jonker}
}

@inproceedings{zhang2021multi,
  title={Multi-agent reinforcement learning: {A} selective overview of theories and algorithms},
  author={Zhang, Kaiqing and Yang, Zhuoran and Basar, Tamer},
  booktitle={Handbook of Reinforcement Learning and Control},
  pages={321--384},
  year={2021},
  publisher={Springer}
}

@article{watkins1992q,
  title={Q-learning},
  author={Watkins, Christopher and Dayan, Peter},
  journal={Machine Learning},
  volume={8},
  pages={279--292},
  year={1992}
}

@article{monderer1996potential,
  title={Potential games},
  author={Monderer, Dov and Shapley, Lloyd},
  journal={Games and Economic Behavior},
  volume={14},
  number={1},
  pages={124--143},
  year={1996}
}

@article{tuyls2006evolutionary,
  title={An evolutionary dynamical analysis of multi-agent learning in iterated games},
  author={Tuyls, Karl and Hoen, Pieter Jan and Vanschoenwinkel, Bram},
  journal={Autonomous Agents and Multi-Agent Systems},
  volume={12},
  pages={115--153},
  year={2006}
}

@article{cressman2014replicator,
  title={The replicator equation and other game dynamics},
  author={Cressman, Ross and Tao, Yi},
  journal={Proceedings of the National Academy of Sciences},
  volume={111},
  pages={10810--10817},
  year={2014}
}

@article{kianercy2012dynamics,
  title={Dynamics of {B}oltzmann {Q} learning in two-player two-action games},
  author={Kianercy, Ardeshir and Galstyan, Aram},
  journal={Physical Review E},
  volume={85},
  number={4},
  year={2012}
}

\end{document}